\documentclass[11pt]{article}
\usepackage[T1]{fontenc}
\usepackage{lmodern}
\usepackage{amsmath,amssymb,amsthm}
\usepackage{booktabs,array}
\usepackage{enumitem}
\usepackage{microtype}
\usepackage{xurl}
\usepackage{tikz}
\usetikzlibrary{positioning,arrows.meta,calc}
\usepackage[hidelinks]{hyperref}
\hypersetup{pdftitle={Finite-ring obstructions for quadratic binary radius-two cellular automata}}
\usepackage[backend=biber,style=numeric,sorting=nyt,sortcites=true]{biblatex}
\usepackage[a4paper,margin=27mm]{geometry}
\selectfont
\newtheorem{theorem}{Theorem}[section]
\newtheorem{proposition}[theorem]{Proposition}
\newtheorem{lemma}[theorem]{Lemma}
\newtheorem{corollary}[theorem]{Corollary}
\newtheorem*{maintheorem}{Main theorem}

\title{Finite-ring obstructions for quadratic binary radius-two cellular automata}
\author{
Houqiao Fu\\
School of Cyber Science and Technology, Hubei University\\
Wuhan, China\\
\texttt{202431120012024@stu.hubu.edu.cn}
}
\date{}

\begin{document}
\maketitle
\begin{abstract}
We study one-dimensional binary cellular automata with a five-slot radius-two local rule of exact algebraic-normal-form degree two, acting on periodic rings of length \(n\).  We prove that every such rule is non-injective whenever \(4\mid n\) and \(n\ge 8\).  The proof begins with the four-cell collapse, where the two extreme formal slots coincide.  A structural classification of the resulting four-variable maps separates the \(65\,472\) exactly quadratic rules into \(63\,456\) rules with an immediate ring-four collision and \(2\,016\) exceptional lifts.  The latter split into layers of sizes \(480\) and \(1\,536\).  Their remaining finite obligations are represented by \(136\) parameter-region constructions and \(768\) per-lift records, respectively.  Each certificate supplies differentiating closed walks of lengths \(8\) and \(12\) with a common pair-graph base vertex.  Concatenation then gives lengths \(8a+12b\), which are exactly the multiples of four from eight onward.  The load-bearing finite certificate core therefore contains \(904=136+768\) independently replayable objects checked by standalone, non-searching programs.  The complete checker CLIs additionally reconstruct expected populations and execute coverage, complement, and regression/guard checks; \(904\) is not a count of total checker operations.  Periodic extension also yields full-shift non-injectivity; that consequence is used here only as a corollary.
\end{abstract}

\section{Introduction}
\label{sec:introduction}

Reversibility on finite periodic configurations is a distinct problem for one-dimensional cellular automata.  Fixing a local rule does not fix a single finite global map: the map depends on the ring length, and short rings can identify formal neighbourhood slots that are distinct on the full shift.  We therefore keep the local encoding, ring-size dependence, and full-shift question separate.

Graph methods for one-dimensional cellular automata provide a standard language for these questions.  De Bruijn representations and related graph constructions have long been used to encode global constraints~\cite{Sutner1991}.  Nobe and Yura studied reversibility under periodic boundary conditions~\cite{NobeYura2004}.  Wang et al. give a graph-based method that computes the complete periodic reversibility sequence for a fixed one-dimensional finite cellular automaton; in particular, their Theorem~3 propagates irreversibility along nonnegative integer combinations of elementary-circuit lengths through a common negative vertex, and their circuit graph records circuit intersections~\cite{WangEtAl2025}.  Haugland and Omland approach the problem through liftings: their finite censuses include diameter-five degree-two data at selected ring sizes, and their proper-lifting classifications exclude degree-two proper liftings in the relevant small-diameter range~\cite{HauglandOmlandAlmost,HauglandOmlandNewClasses}.  Earlier, Omland and St\u{a}nic\u{a} defined finite truncations \(\operatorname{inv}_m(f)=\operatorname{inv}(f)\cap\{k,\ldots,m\}\) and in Section~9 tabulated five-variable quadratic liftings through \(m=15\), including equivalence-class data for \((5,n)\)-liftings with \(7\le n\le15\) and a complete list, up to essential equivalence, of nonlinear quadratic \((5,9)\)-liftings~\cite{OmlandStanica2022}.

These results address complementary questions.  A fixed-rule graph computation can determine an entire reversibility sequence, so the cycle-length combination principle is prior methodology.  Finite lifting tables give exact information on bounded ranges, while proper-lifting classifications concern the induced shift-invariant transformation.  The sources located for this manuscript do not give the specific family-wide statement for every exact-quadratic radius-two rule on the complete congruence class considered here together with the structural four-cell reduction and finite certificate compression used below.  We therefore treat this positioning as a novelty candidate rather than a priority claim.

\begin{maintheorem}
Let \(f:\mathbb F_2^5\to\mathbb F_2\) have ANF degree exactly two.  If
\[
      4\mid n,\qquad n\ge8,
\]
then the periodic global map \(F_{f,n}:\mathbb F_2^n\to\mathbb F_2^n\) is not injective.
\end{maintheorem}

There are \(65\,472\) exactly quadratic five-slot rules.  The proof first classifies which rules can still be bijective after the four-cell slot collapse.  This leaves \(2\,016\) exceptional lifts and gives a ring-four collision for the other \(63\,456\).  The exceptional population splits as \(480+1\,536\).  Its load-bearing finite certificate core consists of \(136\) region-level constructions in the first layer and \(768\) per-representative certificates in the second; the verifier CLIs perform additional reconstruction and regression checks beyond those \(904\) objects.  Thus the proof is not an exhaustive scan over all rules and all ring lengths: structural reductions remove the large easy population, and a graph-theoretic gluing lemma turns two certified lengths into the required infinite congruence class.

Figure~\ref{fig:proof-architecture} summarizes the proof boundary.  The four-cell classification and population arithmetic are structural.  Only the two exceptional branches meet the finite certificate layer.  Their common-base length-8 and length-12 walks are then concatenated, and the elementary semigroup identity \(8a+12b=4(2a+3b)\) supplies every multiple of four from eight onward.  The full-shift non-injectivity consequence follows by periodic extension; its statement is already covered by the proper-lifting literature and is not the organizing claim of this paper.

We organize the manuscript around three contributions:
\begin{enumerate}[label=(\roman*),leftmargin=2.0em,itemsep=0.2em,topsep=0.4em]
  \item the finite-ring non-injectivity theorem for every \(n\ge8\) divisible by four;
  \item the structural four-cell reduction, including the exact split \(65\,472=63\,456+480+1\,536\);
  \item a computer-assisted closure of the two exceptional layers by independently replayable finite certificates, followed by an analytic common-base gluing argument.
\end{enumerate}

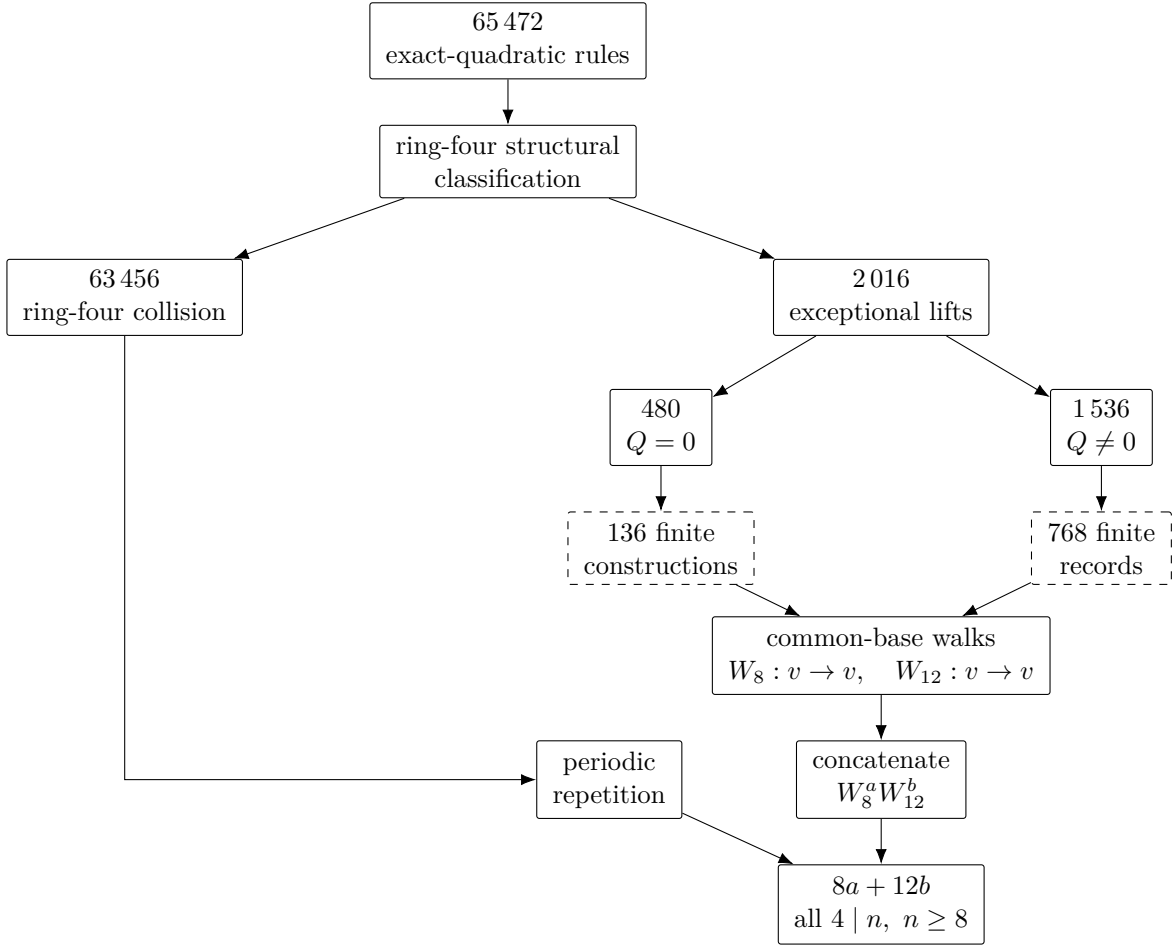
\begin{figure}[t]
  \centering
  \begin{tikzpicture}[
  font=\small,
  node distance=6mm and 9mm,
  box/.style={draw,rounded corners=1pt,align=center,inner xsep=6pt,inner ysep=4pt},
  cert/.style={box,dashed},
  arr/.style={-{Latex[length=2mm]},thin},
  every node/.style={text=black}
]
\node[box] (all) {\(65\,472\)\\exact-quadratic rules};
\node[box,below=of all] (r4) {ring-four structural\\classification};
\node[box,below left=8mm and 18mm of r4] (easy) {\(63\,456\)\\ring-four collision};
\node[box,below right=8mm and 18mm of r4] (exc) {\(2\,016\)\\exceptional lifts};
\node[box,below left=7mm and 8mm of exc] (q0) {\(480\)\\\(Q=0\)};
\node[box,below right=7mm and 8mm of exc] (qn) {\(1\,536\)\\\(Q\ne0\)};
\node[cert,below=of q0] (c0) {\(136\) finite\\constructions};
\node[cert,below=of qn] (cn) {\(768\) finite\\records};
\node[box,below=9mm of $(c0)!0.5!(cn)$] (walks) {common-base walks\\\(W_8:v\to v,\quad W_{12}:v\to v\)};
\node[box,below=of walks] (glue) {concatenate\\\(W_8^aW_{12}^b\)};
\node[box,below=of glue] (semigroup) {\(8a+12b\)\\all \(4\mid n,\ n\ge8\)};
\node[box,left=15mm of glue] (repeat) {periodic\\repetition};

\draw[arr] (all) -- (r4);
\draw[arr] (r4) -- (easy);
\draw[arr] (r4) -- (exc);
\draw[arr] (exc) -- (q0);
\draw[arr] (exc) -- (qn);
\draw[arr] (q0) -- (c0);
\draw[arr] (qn) -- (cn);
\draw[arr] (c0) -- (walks);
\draw[arr] (cn) -- (walks);
\draw[arr] (walks) -- (glue);
\draw[arr] (glue) -- (semigroup);
\draw[arr] (easy) |- (repeat);
\draw[arr] (repeat) -- (semigroup);
\end{tikzpicture}
  \caption{Proof architecture.  Ordinary solid nodes and arrows denote structural or elementary deductions.  Dashed certificate nodes are the only load-bearing finite machine layer: \(136\) constructions for \(Q=0\) and \(768\) records for \(Q\ne0\).  The historical full-rule archive and exploratory searches do not occur in this dependency graph.}
  \label{fig:proof-architecture}
\end{figure}

The remainder of the paper keeps the analytic and finite-certificate parts explicit.  Sections~\ref{sec:preliminaries}--\ref{sec:certificate-reduction} establish the notation, collision-graph language, four-cell classification, and population split.  Sections~\ref{sec:q0-layer} and~\ref{sec:qnonzero-layer} treat the two exceptional layers.  Section~\ref{sec:gluing} proves the all-length theorem, Section~\ref{sec:verification} states the trusted proof boundary, and the appendices record reproducibility metadata and the non-load-bearing ring-eight witness formulas.

\section{Preliminaries and conventions}
\label{sec:preliminaries}

\subsection{Local rules and finite periodic rings}

All algebra is over the binary field \(\mathbb F_2\).  We work with a
one-dimensional binary cellular automaton in a fixed radius-two, five-slot
encoding.  Thus a local rule is a Boolean function
\[
    f:\mathbb F_2^5\longrightarrow \mathbb F_2,
    \qquad
    (x_0,x_1,x_2,x_3,x_4)\longmapsto f(x_0,x_1,x_2,x_3,x_4),
\]
where the five formal slots correspond, in order, to offsets
\(-2,-1,0,1,2\).  For a positive integer \(n\), the associated global map on
the periodic ring of length \(n\) is
\begin{equation}
\label{eq:finite-global-map}
    F_{f,n}:\mathbb F_2^n\longrightarrow\mathbb F_2^n,
    \qquad
    (F_{f,n}(x))_i
    =f(x_{i-2},x_{i-1},x_i,x_{i+1},x_{i+2}),
\end{equation}
with indices read modulo \(n\).

The five slots in \eqref{eq:finite-global-map} are formal positions of the
local rule.  When \(n<5\), distinct formal slots may refer to the same cell of
the ring.  Our convention is to evaluate the original five-variable Boolean
function on those repeated cell values.  We do not first replace \(f\) by a
lower-arity truth table.  At \(n=4\), for example, slots \(x_0\) and \(x_4\)
coincide.

A rule may have effective diameter smaller than five if it does not depend on
one or both extreme formal variables.  We keep effective diameter separate from
the five-slot encoding.  This distinction is relevant when comparing our
radius-two family with the lifting literature, where diameter is normally the
smallest interval of variables on which the rule actually depends.

\subsection{Algebraic normal form and the exact-quadratic family}

Every Boolean local rule has a unique algebraic normal form (ANF)
\[
  f(x)=c+\sum_{i=0}^{4} \ell_i x_i
      +\sum_{0\le i<j\le4} q_{ij}x_ix_j
      +\text{terms of degree at least three}.
\]
We call \(f\) \emph{exactly quadratic} when its ANF degree is exactly two.
Thus ``quadratic'' in the main theorem never means ``degree at most two.''
There are ten quadratic monomials in five variables, so the number of exactly
quadratic five-slot rules is
\begin{equation}
\label{eq:quadratic-population}
    (2^{10}-1)2^6=65\,472.
\end{equation}
This counts ANF coefficient vectors, not standard truth-table rule numbers.
When numerical masks occur in the supplementary material, their encoding is
specified explicitly.

\subsection{Finite rings and the full shift}

The central result of this paper concerns finite periodic rings.  Since
\(F_{f,n}\) is a self-map of the finite set \(\mathbb F_2^n\), injectivity and
bijectivity are equivalent.  The theorem proved in Section~\ref{sec:gluing} is
\begin{equation}
\label{eq:main-theorem-preview}
    \deg_{\mathrm{ANF}}f=2,\qquad 4\mid n,\qquad n\ge8
    \quad\Longrightarrow\quad
    F_{f,n}\text{ is not injective}.
\end{equation}

The bi-infinite cellular automaton is
\[
    F_{f,\infty}:\mathbb F_2^{\mathbb Z}\longrightarrow
    \mathbb F_2^{\mathbb Z},
    \qquad
    (F_{f,\infty}(x))_i
       =f(x_{i-2},x_{i-1},x_i,x_{i+1},x_{i+2}).
\]
A collision on a finite periodic ring lifts periodically to a collision of
\(F_{f,\infty}\).  Hence \eqref{eq:main-theorem-preview} implies full-shift
non-injectivity.  We use that implication only as a corollary.  The statement
that no degree-two proper lifting of effective diameter at most five exists is
already contained in the classifications of Haugland and
Omland~\cite{HauglandOmlandAlmost,HauglandOmlandNewClasses}; our finite-ring
theorem records the stronger uniform obstruction on all ring lengths
\(n\equiv0\pmod4\), \(n\ge8\).

\section{Collision graphs for periodic rings}
\label{sec:collision-graphs}

Graph representations of one-dimensional cellular automata are classical.
De Bruijn methods have been used to study reversibility and related global
properties, including periodic-boundary reversibility; see, for example,
Sutner~\cite{Sutner1991}, Nobe and Yura~\cite{NobeYura2004}, and the recent
periodic reversibility framework of Wang et al.~\cite{WangEtAl2025}.  We use a
pair graph only as a proof language for finite collisions; the graph technique
itself is prior methodology.

\subsection{The pair graph}

Let \(\mathcal C=\mathbb F_2^4\).  A pair-graph vertex is an ordered pair
\((u,v)\in\mathcal C^2\), with
\(u=(u_0,u_1,u_2,u_3)\) and \(v=(v_0,v_1,v_2,v_3)\).  For appended bits
\(a,b\in\mathbb F_2\), let
\[
  \sigma_a(u)=(u_1,u_2,u_3,a),\qquad
  \sigma_b(v)=(v_1,v_2,v_3,b).
\]
There is an edge
\begin{equation}
\label{eq:pair-edge}
   (u,v)\longrightarrow(\sigma_a(u),\sigma_b(v))
\end{equation}
labelled \((a,b)\) precisely when
\begin{equation}
\label{eq:equal-output-edge}
   f(u_0,u_1,u_2,u_3,a)=f(v_0,v_1,v_2,v_3,b).
\end{equation}
An edge is \emph{differentiating} if its label has \(a\ne b\).

\begin{proposition}[Closed-walk collision criterion]
\label{prop:closed-walk-collision}
For \(n\ge5\), the finite map \(F_{f,n}\) is non-injective if and only if the
pair graph contains a closed walk of length \(n\) whose two coordinate words
are distinct.  Equivalently, the walk may be required to contain a
differentiating edge.
\end{proposition}

\begin{proof}
Given distinct \(x,y\in\mathbb F_2^n\) with
\(F_{f,n}(x)=F_{f,n}(y)\), slide a four-cell context around the two words in
parallel.  Equality of the output coordinates gives
\eqref{eq:equal-output-edge} at each step, and periodicity closes the walk after
\(n\) steps.  Conversely, the labels of a closed walk, together with its base
context, reconstruct two periodic words; every edge supplies one equality of
local outputs.  Closure ensures the last four appended labels reproduce the
base context, and a differentiating edge ensures the reconstructed words are
distinct.
\end{proof}

The proof below uses this graph only at lengths 8 and 12.  Ring length four is
handled directly by the slot-coincidence convention of
Section~\ref{sec:preliminaries}.

\subsection{Finite certificates}

A \emph{finite collision certificate} is an explicit finite object from which
one can reconstruct a nontrivial closed walk, or equivalently a pair of
colliding finite configurations, without searching for that witness.  This
separates discovery from verification.  The exploratory programs used during
the project are not theorem premises.  The load-bearing machine layer consists
only of explicit exhibits checked by standalone programs that reconstruct the
relevant rule and verify the stated local equalities directly.

\section{The four-cell collapse}
\label{sec:ring4}

On a four-cell ring, the first and fifth formal slots coincide.  Define the
collapse
\begin{equation}
\label{eq:ring4-collapse}
   \rho(f)(a,b,c,d)=f(a,b,c,d,a).
\end{equation}
Write \(g=\rho(f)\), and let
\(G_4:\mathbb F_2^4\to\mathbb F_2^4\) be its cyclic global map.  We first
record the algebra of the collapse and then classify exactly when \(G_4\) is a
permutation.

\subsection{The collapse map and its kernel}

Let
\[
 R_5=\mathbb F_2[x_0,\ldots,x_4]/(x_i^2+x_i),\qquad
 R_4=\mathbb F_2[x_0,\ldots,x_3]/(x_i^2+x_i),
\]
and let \(V_2\subset R_5\) and \(W_2\subset R_4\) be their subspaces of
Boolean functions of ANF degree at most two.

\begin{lemma}[Collapse kernel]
\label{lem:collapse-kernel}
The substitution \(x_4\mapsto x_0\) restricts to a surjective linear map
\[
    \rho:V_2\longrightarrow W_2.
\]
Its image has dimension \(11\), its kernel has dimension \(5\), and
\begin{equation}
\label{eq:kernel-factorization}
 \ker\rho=(x_0+x_4)\operatorname{span}_{\mathbb F_2}
          \{1,x_0,x_1,x_2,x_3\}.
\end{equation}
Consequently every fixed collapsed rule \(g\in W_2\) has exactly \(2^5=32\)
degree-at-most-two five-slot lifts.
\end{lemma}

\begin{proof}
The map is evaluation under the Boolean-ring substitution \(x_4=x_0\), hence
is linear.  It is surjective because every degree-at-most-two polynomial in
\(x_0,\ldots,x_3\) is the image of the same polynomial viewed as an element of
\(V_2\).  Therefore
\[
   \dim V_2=1+5+\binom52=16,
   \qquad
   \dim W_2=1+4+\binom42=11,
\]
so \(\dim\ker\rho=5\).

Set \(e=x_0+x_4\).  Each of
\[
   e,\quad ex_0,\quad ex_1,\quad ex_2,\quad ex_3
\]
vanishes after the substitution \(x_4=x_0\), so the right-hand side of
\eqref{eq:kernel-factorization} lies in \(\ker\rho\).  These five functions
are independent: in their ANF expansions, the monomials
\(x_4,x_0x_4,x_1x_4,x_2x_4,x_3x_4\), respectively, occur in only one of the
five displayed generators.  Thus the displayed subspace has dimension five
and hence equals the kernel.  Every fibre of a surjective linear map is a
coset of the kernel and therefore has \(32\) elements.
\end{proof}

The exhaustive collapse atlas in the project repository independently reproduces
these dimensions and fibres, but it is regression only and is not used in the
proof.

\subsection{Structural classification on four cells}

Write the collapsed ANF as
\begin{equation}
\label{eq:collapsed-anf}
    g=e+L+Q,
\end{equation}
where
\[
    L=l_0x_0+l_1x_1+l_2x_2+l_3x_3,
    \qquad
    \lambda=l_0+l_1+l_2+l_3,
\]
and \(Q=\sum_{p<r}q_{pr}x_px_r\) is homogeneous quadratic.  Let
\begin{equation}
\label{eq:W-def}
\begin{split}
 W=\{&0,\\
 &(x_0+x_1)(x_2+x_3),\\
 &(x_0+x_2)(x_1+x_3),\\
 &(x_0+x_3)(x_1+x_2)\}.
\end{split}
\end{equation}
The three nonzero members correspond to the three Hamilton cycles of
\(K_4\).

\begin{theorem}[Ring-four structural classification]
\label{thm:ring4-classification}
The map \(G_4\) is bijective if and only if
\begin{equation}
\label{eq:ring4-criterion}
      \lambda=1\qquad\text{and}\qquad Q\in W.
\end{equation}
The constant coefficient \(e\) is unrestricted.
\end{theorem}

\begin{proof}
Let \(\mathbf 1=(1,1,1,1)\),
\(\Pi(x)=x_0+x_1+x_2+x_3\), and set
\[
\begin{aligned}
 u&=(x_0+x_1)(x_2+x_3),\\
 v&=(x_0+x_3)(x_1+x_2),\\
 w&=(x_0+x_2)(x_1+x_3).
\end{aligned}
\]
If
\(\alpha=q_{01}+q_{03}+q_{12}+q_{23}\), summing the four output coordinates
of \(G_4\) gives the identity
\begin{equation}
\label{eq:parity-identity}
   \Pi(G_4(x))=\lambda\Pi(x)+\alpha w(x).
\end{equation}
Indeed, each linear slot contributes one copy of \(\Pi(x)\); quadratic pairs
at cyclic distance two cancel in the four-cell sum, whereas the four pairs at
cyclic distance one or three contribute \(w\).  The four Boolean functions
\(0,w,\Pi,\Pi+w\) have weights \(0,4,8,12\), respectively.  If \(G_4\) is a
permutation then \(\Pi\circ G_4\) has the same weight as \(\Pi\), so
\(\lambda=1\) and \(\alpha=0\).

Let \(E=\{\{p,r\}:q_{pr}=1\}\),
\(m=|E|\bmod2\), and let
\(d_p=\deg_E(p)\bmod2\).  If \(C_d\) denotes the circulant linear map
\[
   (C_d x)_i=\sum_{p=0}^3 d_p x_{i+p-2},
\]
then complementing all four input cells yields
\begin{equation}
\label{eq:complement-derivative}
  G_4(x+\mathbf1)+G_4(x)
   =(\lambda+m)\mathbf1+C_d x.
\end{equation}
With \(\lambda=1\), putting \(x=0\) shows that a permutation must have
\(m=0\), since \(m=1\) would give
\(G_4(\mathbf1)=G_4(0)\).

Also \(\sum_p d_p=0\), so \(d\) has even weight.  Identify
\(\mathbb F_2^4\) with
\(R=\mathbb F_2[z]/(z^4-1)=\mathbb F_2[z]/((z+1)^4)\).  The map \(C_d\) is
multiplication by \(D(z)=\sum_p d_p z^{p-2}\), up to a unit power of \(z\).
For nonzero even-weight \(d\), the \((z+1)\)-adic order of \(D\) is one of
\(1,2,3\), so its image contains
\(1+z+z^2+z^3=(z+1)^3\), the vector \(\mathbf1\).  Thus, if \(d\ne0\), there
is \(x\) with \(C_dx=\mathbf1\), and
\eqref{eq:complement-derivative} gives
\(G_4(x+\mathbf1)=G_4(x)\), a contradiction.  Hence \(d=0\).

The condition \(d=0\) means every vertex of the support graph \((\{0,1,2,3\},E)\)
has even degree.  The Eulerian subgraphs of \(K_4\) are the empty graph, the
four triangles, and the three Hamilton cycles.  The additional condition
\(m=0\) excludes the triangles, which have three edges.  Thus \(Q\in W\).
This proves necessity.

For sufficiency, let \(\sigma\) be cyclic shift and write the linear part of
the global map as
\[
    A=\sum_{p=0}^3 l_p\sigma^{p-2}.
\]
In the same chain ring \(R\), \(A\) is multiplication by a Laurent polynomial
whose value at \(z=1\) is \(\lambda\).  Hence \(A\) is invertible exactly when
\(\lambda=1\).

Let
\[
   U=\{x:\sigma^2x=x\}
     =\{(a,b,a,b):a,b\in\mathbb F_2\}.
\]
For the three nonzero forms in \(W\), the quadratic global maps are,
respectively,
\[
   H_1=(u,v,u,v),\qquad
   H_2=(w,w,w,w),\qquad
   H_3=(v,u,v,u),
\]
and all take values in \(U\).  Put \(P_j=I+H_j\).  Since
\(w\Pi=0\), a direct substitution gives
\[
      P_3P_1=I,\qquad P_2^2=I,
\]
so all three shears are permutations.  On \(U\), the invertible circulant
\(A\) is either the identity or the transposition of the two alternating basis
vectors; accordingly it fixes \(H_2\) and either fixes or exchanges
\(H_1,H_3\).  Therefore, for each nonzero \(Q\in W\), one may choose
\(Q'\in W\setminus\{0\}\) so that
\(A(H_{Q'}(x))=H_Q(x)\) for all \(x\).  Hence
\[
    G_4=T_{e\mathbf1}\circ A\circ P_{Q'},
\]
a composition of bijections.  When \(Q=0\), simply
\(G_4=T_{e\mathbf1}\circ A\).  This proves sufficiency.
\end{proof}

It follows that there are exactly
\begin{equation}
\label{eq:ring4-signatures}
   2\cdot4\cdot8=64
\end{equation}
collapsed permutation signatures: two constants, four choices of \(Q\in W\),
and eight odd-parity linear forms.

\section{Exceptional lifts and the exact population split}
\label{sec:exceptional-lifts}

Lemma~\ref{lem:collapse-kernel} gives 32 degree-at-most-two five-slot lifts of
each collapsed signature.  In the \(Q=0\) case we may choose an affine lift of
the collapsed rule and write every lift as
\begin{equation}
\label{eq:kernel-parametrization}
  f_t=f_0+(x_0+x_4)h_t,
  \qquad
  h_t=t_0+t_1x_0+t_2x_1+t_3x_2+t_4x_3.
\end{equation}
The lift is affine precisely when
\((t_1,t_2,t_3,t_4)=0\), giving two affine values of \(t_0\) and therefore
30 exactly quadratic lifts per \(Q=0\) collapsed signature.  If \(Q\ne0\),
every lift is exactly quadratic because its collapse already has nonzero
quadratic part.

Among the 64 collapsed permutation signatures, 16 have \(Q=0\) and 48 have
\(Q\ne0\).  Consequently the number of exactly quadratic rules surviving the
ring-four obstruction is
\begin{equation}
\label{eq:exceptional-count}
   16\cdot30+48\cdot32=480+1536=2016.
\end{equation}
The complementary population is therefore
\begin{equation}
\label{eq:easy-count}
   65\,472-2\,016=63\,456.
\end{equation}
Every rule counted in \eqref{eq:easy-count} already has a collision on the
four-cell ring.

\begin{corollary}[Population reduction]
\label{cor:population-reduction}
The exactly quadratic five-slot family splits into \(63\,456\) rules with a
ring-four collision and \(2\,016\) exceptional lifts.  The exceptional set is
the disjoint union of a \(Q=0\) layer of size \(480\) and a \(Q\ne0\) layer
of size \(1\,536\).
\end{corollary}

This population arithmetic is derived from
Theorem~\ref{thm:ring4-classification} and
Lemma~\ref{lem:collapse-kernel}.  The historical 65,472-rule full-rule archive is not
a premise of Corollary~\ref{cor:population-reduction}; it is retained only as
an independent regression archive.

\section{Reduction to two finite certificate layers}
\label{sec:certificate-reduction}

After Corollary~\ref{cor:population-reduction}, only the \(2\,016\) exceptional
lifts require witnesses beyond ring four.  The proof treats their two layers
separately.

For \(Q=0\), Section~\ref{sec:q0-layer} reduces the \(480\) rules to eight
constant-free representatives and a finite family of 136 certified
constructions.  Each construction supplies compatible collision walks of
lengths 8 and 12 with a common base vertex.  For \(Q\ne0\),
Section~\ref{sec:qnonzero-layer} uses 768 representative-lift certificates,
each explicitly recording such a common-base pair of walks; output complement
doubles these representatives to all \(1\,536\) rules of that layer.

Thus the load-bearing finite machine core contains exactly
\begin{equation}
\label{eq:machine-core}
      136+768=904
\end{equation}
explicit finite exhibits.  Section~\ref{sec:gluing} turns the two witness
lengths into all multiples of four from eight onward, and
Section~\ref{sec:verification} states the exact trusted-computation boundary.
No collapse-atlas fact and no full-rule archive record is part of this core.

\section{The exceptional layer with \texorpdfstring{$Q=0$}{Q=0}}
\label{sec:q0-layer}

We first treat the exceptional collapsed signatures whose four-variable
quadratic part vanishes.  There are sixteen such signatures: the constant
coefficient is arbitrary and the linear form has odd coefficient sum.  The
output constant can be removed without changing collisions.

\begin{lemma}[Output complement]
\label{lem:output-complement}
Let \(\bar f=f+1\).  For every \(n\ge1\),
\[
        F_{\bar f,n}(x)=F_{f,n}(x)+\mathbf 1_n.
\]
Consequently \(F_{\bar f,n}(x)=F_{\bar f,n}(y)\) if and only if
\(F_{f,n}(x)=F_{f,n}(y)\), and \(F_{\bar f,n}\) is injective if and only if
\(F_{f,n}\) is injective.
\end{lemma}

\begin{proof}
Adding the constant one to the local rule toggles every output coordinate and
nothing else.  Equality of two output vectors is therefore preserved.
\end{proof}

\subsection{Kernel parametrization}

Fix one of the eight odd four-variable linear forms \(L\), take constant
coefficient zero, and choose the affine five-slot representative \(f_0=L\)
that ignores \(x_4\).  By Lemma~\ref{lem:collapse-kernel}, all lifts of the
same collapsed signature have the form
\begin{equation}
\label{eq:q0-param}
  f_t=f_0+(x_0+x_4)h_t,
  \qquad
  h_t=t_0+t_1x_0+t_2x_1+t_3x_2+t_4x_3,
  \qquad t\in\mathbb F_2^5.
\end{equation}
The two parameters with \((t_1,t_2,t_3,t_4)=0\) are affine, while the other
thirty are exactly quadratic.  Thus the eight constant-free representatives
account for \(8\cdot30=240\) exact-quadratic lifts.  Lemma~\ref{lem:output-complement}
doubles them to the complete \(Q=0\) exceptional layer of size \(480\).

\subsection{Parametric collision regions}
\label{subsec:q0-regions}

The finite certificates used below are not the output of a search that must be
trusted at verification time.  Their compression rests on a simple affine
observation.  Fix a candidate pair-graph edge, so the two five-bit input
windows \(z,z'\in\mathbb F_2^5\) are fixed.  By \eqref{eq:q0-param}, the edge
condition is
\begin{equation}
\label{eq:q0-edge-affine}
 f_t(z)+f_t(z')
 =f_0(z)+f_0(z')
   +\sum_{j=0}^{4}t_j\bigl(k_j(z)+k_j(z')\bigr)=0,
\end{equation}
where \(k_0,\ldots,k_4\) are the five kernel generators in
\eqref{eq:kernel-factorization}.  Hence the set of parameters on which a fixed
edge is legal is an affine subspace (possibly empty) of \(\mathbb F_2^5\).
The same is true for a fixed finite walk, by intersecting its edge conditions.

For the shipped construction one starts with a collision pair
\((x,y)\in(\mathbb F_2^8)^2\) and appends four-bit blocks \(u,v\) to obtain
\(x\Vert u,y\Vert v\in\mathbb F_2^{12}\).  The four-cell contexts at the
chosen start position are inherited from \(x,y\); consequently the associated
length-8 and length-12 pair-graph walks have the same base vertex.  For fixed
\((x,y,u,v)\), legality of both walks is again a system of affine equations in
\(t\), and therefore defines an explicitly checkable parameter region.

The frozen $Q=0$ certificate artifact stores 136 such region-level constructions.  Regions
may overlap; what matters is exact coverage of the thirty quadratic parameters
for each of the eight representatives.

\begin{proposition}[Finite certificates for the \(Q=0\) layer]
\label{prop:q0-certificates}
For each of the eight constant-free odd-linear representatives and every one
of its thirty exactly quadratic kernel lifts, at least one of the 136 shipped
certificate constructions applies.  Every applicable construction gives a
closed differentiating walk of length \(8\) and a closed differentiating walk
of length \(12\) in the pair graph of the lift, and the two walks have a common
base vertex.  The union of the certified regions for each representative is
exactly its set of thirty exactly quadratic parameters.
\end{proposition}

\begin{proof}[Finite verification]
The certificate archive records, for every construction, the representative,
the parameter region, the ring-8 words \((x,y)\), the appended four-bit blocks
\((u,v)\), and the resulting ring-12 words.  The standalone checker
\texttt{verifier/check\_q12i\_joinability.py} independently generates the
eight odd-parity labels \(L\), constructs from each label the canonical
constant-free affine base, reconstructs the fixed five-dimensional collapse
kernel and all \(32\) lifts, and requires the archive family to equal that
canonical lift list.  It checks the \(30/2\) quadratic--affine split and exact
set equality between the mathematically expected \(240\) constant-free
quadratic masks and the archive's \(240\) distinct physical masks, rejecting
cross-family duplicates, omissions, unexpected masks, and label/base
mismatches.  It also solves the printed region conditions over the parameter
cube; replays every recorded ring-8 and ring-12 walk and checks closure,
legality of every local edge, equality of the two ring images, and
differentiation; and checks coverage in both directions.  It does
not rank candidates, search for witnesses, or construct a new cover.  Thus the
machine-dependent content of the proposition is only the verification of the
136 explicit finite exhibits and their stated parameter regions.
\end{proof}

For orientation, one certificate in the archive has
\[
  L=x_3,\qquad
  t\in\{10,13,16,23\},\qquad
  (x,y)=(37,38),\qquad (u,v)=(2,2),
\]
with the integers read in the bit convention of the supplementary archive.
Its two stored words are \(549\) and \(550\) at length twelve.  This example
is illustrative only; Proposition~\ref{prop:q0-certificates} depends on the
whole checked archive, not on the example.

\section{The exceptional layer with \texorpdfstring{$Q\ne0$}{Q nonzero}}
\label{sec:qnonzero-layer}

The remaining collapsed signatures have one of the three nonzero Hamilton
quadratic forms in \(W\).  After fixing constant coefficient zero, a
representative is
\begin{equation}
\label{eq:qnonzero-base}
       f_0=L+Q,
\end{equation}
where \(L\) is one of the eight odd linear forms and \(Q\) is one of the
three Hamilton supports
\[
 (0,1,2,3),\qquad (0,1,3,2),\qquad (0,2,1,3),
\]
interpreted as four-cycle edge sets.  We index these supports by
\(q=0,1,2\), respectively.  Hence there are \(3\cdot8=24\) constant-free
representatives.  Every one of their 32 kernel lifts is exactly quadratic,
because the collapsed quadratic part is already nonzero.  This gives
\(24\cdot32=768\) representative lifts, and Lemma~\ref{lem:output-complement}
doubles them to the complete \(Q\ne0\) layer of size \(1536\).

\subsection{Quadratic-base reduction at ring eight}
\label{subsec:qnonzero-ring8}

The frozen hyperplane lemmas give a structural explanation for the ring-eight
collisions in this layer.  This material is explanatory rather than
load-bearing for the all-length result: the common-base certificates in the
next subsection contain their own length-8 walks.  We retain the reduction
because it shows why the finite layer is organized by elementary parity
conditions.

Write \(y=x+d\), let \(W_i(x)\) be the five-slot window at cell \(i\), and
let \(D=W_i(d)\).  If the base rule is \(c+L+Q\), polarization gives
\begin{equation}
\label{eq:GQ}
\begin{split}
 F_t(x)_i+F_t(x+d)_i
   ={}&L(D)+Q(D)+B_Q(W_i(x),D)\\
      &+\delta_i\,\ell_t(D_0,D_1,D_2,D_3)
       +\varepsilon_i\,h_t(W_i(x)+D),
\end{split}
\end{equation}
where
\[
 B_Q(z,w)=Q(z+w)+Q(z)+Q(w),\qquad
 \ell_t=t_1x_0+t_2x_1+t_3x_2+t_4x_3,
\]
\[
 \delta_i=x_{i-2}+x_{i+2},\qquad
 \varepsilon_i=d_{i-2}+d_{i+2}.
\]
For fixed defect \(d\) and kernel parameter \(t\), the right-hand side is
affine in the background window.  The ring-eight witness problem therefore
reduces to a small linear system rather than a nonlinear search over
backgrounds.

One representative case illustrates the mechanism.  For the four-periodic
defect \(A=00110011\), the first Hamilton support \(q=0\), and
\(p=l_1+l_3=1\), set
\[
  \alpha=t_3+t_4,\qquad \gamma=t_1+t_2.
\]
On the hyperplane \(t_2+t_4=0\), the choice
\[
   s=1+l_0+l_1,\qquad
   (u_0,u_1,u_2,u_3)=(1+\gamma,\gamma,1+\alpha,\alpha)
\]
solves the reduced ring-eight equations, where
\(u_j=x_j+x_{j+4}\) and \(s=x_0+x_1+x_2+x_3\).  On the complementary
hyperplane \(t_2+t_4=1\), a second defect family supplies an explicit
background.  The other parity cases are analogous: for \(p=0\) the split is
by \(t_1+t_3\), while the remaining two Hamilton supports split by
\(t_1+t_2+t_3+t_4\).  Appendix~\ref{app:ring8-witnesses} records the full
reduced system and all six explicit background formulas.  None of those
formulas is used as a premise of Proposition~\ref{prop:qnonzero-certificates}.

\subsection{Common-base certificates}

The all-length result for this layer uses a different finite object.  The
common-base certificate archive contains one record for every pair consisting of a constant-free
representative and a kernel parameter.  Each record stores one base vertex and
two explicit closed differentiating walks from that vertex, of lengths 8 and
12.

\begin{proposition}[Finite certificates for the \(Q\ne0\) layer]
\label{prop:qnonzero-certificates}
For each of the 768 constant-free representative lifts in the \(Q\ne0\)
exceptional layer, the shipped certificate archive contains a closed
differentiating length-8 walk and a closed differentiating length-12 walk with
the same base vertex.  The archive covers the expected 768
\((\text{representative},t)\) pairs exactly once.  By output complement, the
same collision statements hold for all \(1536\) rules of the layer.
\end{proposition}

\begin{proof}[Finite verification]
The standalone checker \texttt{verifier/check\_q12jc\_joinability.py} imports
no project module and performs no witness search.  It rebuilds the local truth
table and ANF degree, the ring-8 and ring-12 maps, the pair graph, the ring-four
collapse population, the 24 constant-free representatives and their 32 kernel
lifts, and the owner partition used only as certificate metadata.  For each
record and each of its two walks it verifies the vertex evolution from the
labels, closure, legality of every local edge, the exact differentiating
positions, reconstruction of the two ring words, distinctness, and equality of
their global images.  It then checks exact coverage of all 768 expected pairs
and verifies the output-complement identity used for the doubling.  The
certificate records, rather than any earlier search or fitted-region data, are
the finite premise of the proposition.
\end{proof}

A typical archive record contains the representative data \((Q,L,t)\), a base
pair-state, and two walk objects of the form
\[
  (\text{labels},\text{vertices},\text{two words},
    \text{differentiating positions}).
\]
For example, the first shipped record has base vertex \(0\), \(q=0\),
\(L=x_3\), \(t=0\); its ring-8 words are \((16,144)\), and the recorded
length-8 walk has its differentiating position at index \(3\).  The complete
archive is supplementary data; no list of 768 records is reproduced here.

\section{Common-base gluing and the finite-ring theorem}
\label{sec:gluing}

We now remove the finite lengths 8 and 12 from the statement.  This step is
purely graph-theoretic and arithmetic.

\begin{lemma}[Common-base gluing]
\label{lem:common-base-gluing}
Let a pair graph contain a closed differentiating walk
\(W_8:v\to v\) of length 8 and a closed differentiating walk
\(W_{12}:v\to v\) of length 12.  For any \(a,b\ge0\), not both zero,
\[
       W_8^aW_{12}^b
\]
is a closed differentiating walk at \(v\), of length \(8a+12b\).
\end{lemma}

\begin{proof}
Both walks begin and end at the same vertex, so any concatenation of copies is
again closed at that vertex.  If \(a>0\), the concatenation contains a
differentiating edge belonging to a copy of \(W_8\); if \(a=0\), then
\(b>0\) and it contains one belonging to \(W_{12}\).  Lengths add under
concatenation.
\end{proof}

\begin{lemma}[The \(8,12\) semigroup]
\label{lem:semigroup}
Every integer \(n\ge8\) divisible by four can be written
\(n=8a+12b\) with \(a,b\ge0\).
\end{lemma}

\begin{proof}
Write \(n=4m\), where \(m\ge2\).  If \(m\) is even, take
\[
     a=m/2,\qquad b=0,
\]
so \(n=8(m/2)\).  If \(m\) is odd, then \(m\ge3\), and take
\[
     a=(m-3)/2,\qquad b=1,
\]
so \(n=12+8((m-3)/2)\).
\end{proof}

We also use the elementary repetition of a finite periodic collision.

\begin{lemma}[Finite periodic repetition]
\label{lem:finite-repetition}
If \(F_{f,m}\) has a collision and \(m\mid n\), then
\(F_{f,n}\) has a collision.
\end{lemma}

\begin{proof}
Repeat the two colliding length-\(m\) configurations periodically to length
\(n\).  Every radius-two window in the repeated length-\(n\) words is the
corresponding window modulo \(m\), including the slot coincidences when they
occur.  The repeated words remain distinct and have equal images.
\end{proof}

\begin{theorem}[Finite-ring obstruction]
\label{thm:finite-ring}
Let \(f:\mathbb F_2^5\to\mathbb F_2\) have ANF degree exactly two.  If
\[
      4\mid n,\qquad n\ge8,
\]
then \(F_{f,n}:\mathbb F_2^n\to\mathbb F_2^n\) is not injective.
\end{theorem}

\begin{proof}
There are \(65\,472\) exactly quadratic five-slot rules.  By
Corollary~\ref{cor:population-reduction}, \(63\,456\) of them have a ring-four
collision.  Lemma~\ref{lem:finite-repetition} propagates that collision to
every ring length divisible by four.

The remaining \(2016\) rules are the exceptional lifts.  In the \(Q=0\) layer,
Proposition~\ref{prop:q0-certificates} supplies common-base differentiating
walks of lengths 8 and 12 for the 240 constant-free exactly quadratic lifts,
and Lemma~\ref{lem:output-complement} doubles the conclusion to all 480 rules.
In the \(Q\ne0\) layer,
Proposition~\ref{prop:qnonzero-certificates} supplies the same two lengths for
all 768 constant-free representative lifts, and output complement doubles the
conclusion to all 1536 rules.

For every exceptional rule, Lemma~\ref{lem:common-base-gluing} therefore gives
a differentiating closed walk at each length \(8a+12b\) with \(a,b\ge0\), not
both zero.  Lemma~\ref{lem:semigroup} identifies these lengths with every
multiple of four from eight onward.  Proposition~\ref{prop:closed-walk-collision}
converts each such walk into a nontrivial collision.  The three disjoint
populations
\[
      63\,456+480+1536=65\,472
\]
exhaust the exactly quadratic family, proving the theorem.
\end{proof}

\begin{corollary}[Periodic full-shift collision]
\label{cor:full-shift}
Every exactly quadratic binary five-slot radius-two rule induces a
non-injective map on \(\mathbb F_2^{\mathbb Z}\).  The collision may be chosen
periodic.
\end{corollary}

\begin{proof}
Apply Theorem~\ref{thm:finite-ring}, for example at \(n=8\), and extend the two
colliding ring configurations periodically to \(\mathbb Z\).  Their radius-two
windows and hence their images repeat with the same period.
\end{proof}

The statement of Corollary~\ref{cor:full-shift} is consistent with, and as a
reversibility statement is already covered by, the proper-lifting
classification of Haugland and Omland~\cite{HauglandOmlandAlmost,HauglandOmlandNewClasses}.
The role of Theorem~\ref{thm:finite-ring} is the more specific obstruction on
every periodic ring in the congruence class \(n\equiv0\pmod4\) from \(n=8\)
onward.

\section{Verification architecture and proof boundary}
\label{sec:verification}

Theorem~\ref{thm:finite-ring} is a computer-assisted proof with independently
replayable finite certificates.  This section separates its mathematical
arguments from the finite replay obligations and from computations retained
only for regression.

\subsection{Structural layer}

The following parts of the proof are analytic and do not require trusting an
exhaustive search:
\begin{itemize}
  \item the finite-ring definitions, the pair-graph correspondence, and finite
        periodic repetition;
  \item the collapse lemma, including \(\dim\operatorname{im}\rho=11\),
        \(\dim\ker\rho=5\), and the explicit kernel
        \((x_0+x_4)\operatorname{span}\{1,x_0,x_1,x_2,x_3\}\);
  \item the ring-four classification
        \(G_4\) bijective if and only if \(\lambda=1\) and \(Q\in W\);
  \item the resulting population split
        \(65\,472=63\,456+480+1536\);
  \item the output-complement identity;
  \item common-base concatenation and the elementary semigroup identity
        \(8a+12b=4(2a+3b)\).
\end{itemize}
The explicit ring-eight hyperplane witnesses summarized in
Section~\ref{subsec:qnonzero-ring8} and Appendix~\ref{app:ring8-witnesses}
are structural explanatory results, but the main theorem does not require
them: Proposition~\ref{prop:qnonzero-certificates} independently replays the
length-8 walk contained in every common-base certificate.

\subsection{Finite certificate layer}

After the structural reductions, exactly two load-bearing finite certificate
populations remain:
\begin{equation}
\label{eq:finite-proof-core}
  \underbrace{136}_{Q=0\text{ region constructions}}
  +\underbrace{768}_{Q\ne0\text{ representative records}}
  =904.
\end{equation}
These \(904\) objects are the load-bearing finite exhibits.  They are not a
count of all operations performed by the release checkers: the full CLIs also
reconstruct finite populations and run degree, coverage, complement, and
regression/guard checks.  The exhibits are checked by two standalone,
non-searching programs:
\begin{center}
\begin{tabular}{lll}
\hline
layer & certificate object & standalone checker\\
\hline
\(Q=0\) & 136 parameter-region constructions &
\texttt{check\_q12i\_joinability.py}\\
\(Q\ne0\) & 768 common-base records &
\texttt{check\_q12jc\_joinability.py}\\
\hline
\end{tabular}
\end{center}

The common checker contract is deliberately stronger than accepting stored
summary counts.  Each checker reconstructs the expected finite population from
the mathematical parametrization rather than discovering it from the archive,
reads explicit certificate content, verifies the relevant rule and degree,
checks every stored walk through local rule equalities and closure, checks that
the two reconstructed configurations are distinct with equal global images,
and verifies exact certificate coverage.  For \(Q=0\), this includes the
independent binding \(L\mapsto f_0(L)\mapsto\) the canonical 32-lift family and
global equality of the actual and expected 240-rule ANF-mask sets.  Neither checker searches the pair
graph for a new witness or optimizes a cover.  The infinite family of ring
lengths is then obtained from Lemmas~\ref{lem:common-base-gluing} and
\ref{lem:semigroup}, not from a finite ring scan.

The supplementary release pins the certificate archives and standalone
checkers.  Appendix~\ref{app:artifacts} gives repository-relative paths,
generator identifiers, SHA-256 values, and exact replay commands.  Those
technical identifiers are reproducibility metadata rather than additional
mathematical premises.

\subsection{Certificate schema}

The two archives use slightly different compression levels.  A $Q=0$ certificate record is
a region-level template.  Its essential fields are
\[
 (L,\;\text{parameter region},\;x,y,u,v,\;\text{common base}),
\]
where \((x,y)\) is a ring-8 collision and appending \((u,v)\) produces the
ring-12 collision.  A $Q\ne0$ certificate record is per representative lift and has the
schema
\[
 (Q,L,t,\;v,\;W_8,\;W_{12}),
\]
where each walk stores labels, vertices, reconstructed words, and the indices
of differentiating labels.  These data are sufficient for deterministic
replay; no exploratory trace is needed to understand what is being verified.

\subsection{Regression-only computations}

Several larger computations remain useful as independent consistency checks
but are outside the theorem dependency graph.  In particular, the following
are regression only:
\begin{itemize}
  \item the historical 65,472-rule full-rule archive;
  \item the exhaustive collapse atlas, now superseded in the proof by
        Lemma~\ref{lem:collapse-kernel};
  \item exploratory region scans, fitted-region and rank censuses, template
        searches, and canonical-witness extension statistics from earlier exploratory
        rounds;
  \item finite sweeps that concatenate certificates at a sample of ring
        lengths after the general gluing lemma has already been proved.
\end{itemize}
None of these computations is required to infer
Theorem~\ref{thm:finite-ring} from the structural lemmas and the 904 finite
certificate exhibits.

\section{Discussion and limitations}
\label{sec:discussion}

The theorem concerns local rules of ANF degree exactly two.  Affine rules are not part of the statement, and degree-three or higher rules require different arguments.  The ring-length conclusion is also deliberately partial: odd lengths and lengths congruent to two modulo four are outside the theorem.  Nothing in the proof classifies reversibility on those remaining congruence classes.

The pair graph is a representation device rather than a methodological claim.  Its role is to turn a finite collision into a closed differentiating walk and to make the common-base concatenation transparent.  Graph-based reversibility methods for one-dimensional cellular automata are prior methodology~\cite{Sutner1991,NobeYura2004,WangEtAl2025}.  In particular, Wang et al. compute complete reversibility sequences for fixed rules, and their Theorem~3 already uses nonnegative integer combinations of circuit lengths through a common vertex; neither graph concatenation nor the semigroup-of-cycle-lengths idea is claimed here as new.  What matters for the present result is the family-wide reduction before the finite graph data are used: the four-cell classification removes \(63\,456\) rules analytically, and the remaining \(2\,016\) rules are reduced to two fixed certificate families whose length-8 and length-12 witnesses are sufficient for the stated congruence class.

The relevant distinction is between a fixed-rule complete-spectrum algorithm and a structural statement quantified over an entire rule family.  In the present proof, the rule quantifier is compressed by the ring-four classification and exact population split, while the ring-length quantifier is discharged by periodic repetition for the ring-four branch and by common-base concatenation plus the \(8,12\) semigroup lemma for the exceptional branches.  The machine layer is finite and fixed before those arguments are applied.

The finite layer remains genuine.  The \(904\) load-bearing certificate objects are not replaced by a single symbolic construction, and \(904\) should not be read as the total number of operations executed by the full verifier CLIs: \(136\) parameter-region records cover the \(Q=0\) representatives, and \(768\) records cover the constant-free \(Q\ne0\) representative lifts.  Their checkers rebuild the rules, replay the stored walks and local equalities, verify differentiation and coverage, and do not search for replacement witnesses.  The proof should therefore be read as a structural argument with an independently replayable finite certificate boundary, not as a purely symbolic classification.

Omland and St\u{a}nic\u{a}'s Section~9 is also close enough to require explicit separation: their \(\operatorname{inv}_{15}\) column is the truncation of the lifting spectrum to lengths at most 15, and their Tables~3--4 give finite computational data for five-variable quadratic liftings rather than an all-length theorem for the entire exact-quadratic family~\cite{OmlandStanica2022}.  We do not treat those finite tables as a statement about all \(n\).

Finally, the finite-ring statement and the full-shift consequence have different roles.  A finite periodic collision extends to the bi-infinite shift, so the theorem implies non-injectivity on \(\mathbb F_2^{\mathbb Z}\).  The absence of degree-two proper liftings in the relevant small-diameter regime is already represented in the classifications of Haugland and Omland~\cite{HauglandOmlandAlmost,HauglandOmlandNewClasses}.  The contribution candidate of this manuscript is therefore the uniform finite-ring obstruction on the stated congruence class, not the full-shift non-reversibility statement by itself.

\section{Conclusion}
\label{sec:conclusion}

Every binary radius-two local rule of exact ANF degree two is non-injective on each periodic ring of length \(n\ge8\) divisible by four.  The proof combines a structural classification on four cells, an exact reduction to two exceptional lift layers, a finite core of independently replayable certificates at lengths 8 and 12, and a common-base gluing argument that supplies all remaining lengths.  This separation keeps the infinite part of the theorem analytic while making the finite computer-assisted boundary explicit and reproducible.

\section*{Code and data availability}
The certificate archives, standalone verification scripts, and fresh replay logs supporting the computer-assisted part of this paper are provided directly under \path{anc/} as ancillary files accompanying this arXiv submission.

\appendix

\section{Reproducibility and artifact manifest}
\label{app:artifacts}

The theorem depends on the two finite certificate families described in Section~\ref{sec:verification}.  The Q12i certificate archive is byte-identical to the pre-patch frozen archive; only its verifier contract is patched.  The load-bearing certificate files, standalone checkers, and fresh replay logs are distributed directly in the public ancillary tree under \path{anc/}.  The file \path{anc/MANIFEST_SHA256.txt} records SHA-256 values for the distributed reproducibility files; the core certificate and checker hashes are also pinned below so that the finite objects can be identified independently of archive packaging.

The two load-bearing replays are invoked from the ancillary root (the \path{anc/} directory in the arXiv source package) by
\begin{verbatim}
python verifier/check_q12i_joinability.py \
    --in results/q12i_joinability.json
python verifier/check_q12jc_joinability.py \
    --archive results/q12jc_common_base.jsonl
\end{verbatim}
Both programs use only the Python standard library and reconstruct the finite populations needed for their respective contracts.  Their mathematical obligations are summarized in Section~\ref{sec:verification}; exploratory generators are not required to replay the proof.

Fresh non-quick release replays and negative verifier-regression results are archived under \path{anc/release_logs/}.  Their SHA-256 values, together with the regression-driver hash, are recorded in \path{anc/MANIFEST_SHA256.txt}.  Empty stderr files are intentionally omitted; successful runs are recorded by the corresponding exit-code and stdout files.  The regression mutations test verifier rejection paths and are not theorem certificates.

\begin{table}[ht]
\centering
\small
\begin{tabular}{@{}>{\raggedright\arraybackslash}p{0.25\linewidth}>{\raggedright\arraybackslash}p{0.25\linewidth}>{\raggedright\arraybackslash}p{0.42\linewidth}@{}}
\toprule
object & role & SHA-256 / pinned identifier\\
\midrule
\path{results/q12i_joinability.json} & \(Q=0\) report, 136 constructions & file SHA-256 \path{d98b3d9c90fa4c6154d6fbcad4dfb295f91e6fbfc29bd22a3759bf0e71101ed3}; generator \path{q12i-join-2.9}; report digest \path{9a457be940b42b60790ebc7627a94eb92c9736ad4394161cf5eac1c54983fd3b}\\
\addlinespace
\path{verifier/check_q12i_joinability.py} & standalone \(Q=0\) replay checker & file SHA-256 \path{88d3c0ea003a7228d2fc5b520dd6094112d0764f9430dd330114e12675410f2b}\\
\addlinespace
\path{results/q12jc_common_base.jsonl} & \(Q\ne0\) archive, 768 records & file SHA-256 \path{8b1327f8e70f89ac8b9f3e5688b9bc01335e7b99647ef2466f2c8f49ffecb872}\\
\addlinespace
\path{results/q12jc_joinability.json} & \(Q\ne0\) report & file SHA-256 \path{4fcc6d42bdc262a9d82442234468f98ba152d691a5f867cf763f50104fd91d7a}; generator \path{q12jc-join-2.12}; report digest \path{a6a1aae5d1e712703d5b60be0c9605808d5e5386a3bf078f3d843ab93f570266}\\
\addlinespace
\path{verifier/check_q12jc_joinability.py} & standalone \(Q\ne0\) replay checker & file SHA-256 \path{8dabb815858bbf3bb0c486dd982262102ba5af0ba24932a1fab9c806712f0ae3}\\
\bottomrule
\end{tabular}
\caption{Load-bearing finite artifacts and standalone checkers.  The long digests identify the frozen objects; they are not mathematical assumptions beyond fixing which certificate data are replayed.}
\label{tab:artifact-manifest}
\end{table}

Larger regression artifacts used during development include the historical \(65\,472\)-rule archive, the collapse atlas, fitted-region scans, rank censuses, and finite concatenation sweeps.  These objects can detect drift, but they are not inputs to Theorem~\ref{thm:finite-ring}.

\section{Explicit ring-eight structural witnesses}
\label{app:ring8-witnesses}

This appendix records the explicit ring-eight formulas summarized in Section~\ref{subsec:qnonzero-ring8}.  They explain the parity structure of the \(Q\ne0\) layer but are not a dependency of the common-base certificate proposition or of Theorem~\ref{thm:finite-ring}.

For the four-periodic defect \(A=00110011\) with quadratic support index \(q=0\), let
\[
   u_j=x_j+x_{j+4}\quad(0\le j<4),
   \qquad s=x_0+x_1+x_2+x_3,
\]
and set
\[
 \alpha=t_3+t_4,\quad \beta=t_2+t_3,\quad
 \gamma=t_1+t_2,\quad \eta=t_1+t_4.
\]
The eight cell equalities reduce to
\begin{align}
 s+(1+\alpha)u_2+u_3 &= 1+l_2+l_3,\label{eq:app-qjr-row0}\\
 s+(1+\beta)u_3 &= 1+l_1+l_2,\label{eq:app-qjr-row1}\\
 s+\gamma u_0 &= 1+l_0+l_1,\label{eq:app-qjr-row2}\\
 s+u_0+\eta u_1 &= 1+l_0+l_3,\label{eq:app-qjr-row3}\\
 u_0+u_1+u_2+u_3 &=0.\label{eq:app-qjr-parity}
\end{align}
The frozen source re-derives these rows from the truth table.  With
\(p=l_1+l_3\), \(r=t_1+t_2+t_3+t_4\), and
\(\Delta=t_1+t_4\), explicit backgrounds are given by Table~\ref{tab:ring8-witnesses}.

\begin{table}[ht]
\centering
\small
\begin{tabular}{@{}llll@{}}
\toprule
case & hypothesis & \(s\) & \((u_0,u_1,u_2,u_3)\)\\
\midrule
\(A,q=0,p=1\) & \(t_2+t_4=0\) & \(1+l_0+l_1\) & \((1+\gamma,\gamma,1+\alpha,\alpha)\)\\
\(A,q=0,p=0\) & \(t_1+t_3=0\) & \(1+l_0+l_1+t_1+t_2\) & \((\gamma,1+\Delta,\Delta,1+\gamma)\)\\
\(B,q=0,p=1\) & \(t_2+t_4=1\) & \(0\) & \((0,1,0,1)\)\\
\(B,q=0,p=0\) & \(t_1+t_3=1\) & \(0\) & \((1,0,1,0)\)\\
\(B,q=1,2\) & \(r=0\) & \(1+l_0+l_2+t_1+t_3\) & \((1,1,1,1)\)\\
\(C\) & \(r=1\) & \(0\) & \((1,1,1,1)\)\\
\bottomrule
\end{tabular}
\caption{Explicit ring-eight backgrounds from the frozen hyperplane lemmas.  Here \(B=01010101\) and \(C=11111111\).  The paired hypotheses are complementary in the relevant parameter space.}
\label{tab:ring8-witnesses}
\end{table}

For each row one recovers a background by taking
\[
    x_0=s,\qquad x_1=x_2=x_3=0,\qquad x_{j+4}=x_j+u_j.
\]
The complementary hypotheses cover all four-bit kernel parameters in the indicated families.  These formulas supply an explanatory structural proof of ring-eight non-injectivity for the layer; the all-length theorem instead uses the independently replayed common-base records of Proposition~\ref{prop:qnonzero-certificates}.

\printbibliography

@article{HauglandOmlandAlmost,
  author  = {Haugland, Jan Kristian and Omland, Tron},
  title   = {Shift-invariant transformations and almost liftings},
  journal = {Cryptography and Communications},
  volume  = {18},
  pages   = {705--726},
  year    = {2026},
  doi     = {10.1007/s12095-025-00848-w},
  note    = {Published online 14 November 2025}
}

@article{HauglandOmlandNewClasses,
  author  = {Haugland, Jan Kristian and Omland, Tron},
  title   = {New classes of reversible cellular automata},
  journal = {Designs, Codes and Cryptography},
  volume  = {94},
  eid     = {131},
  year    = {2026},
  doi     = {10.1007/s10623-026-01869-z}
}

@article{WangEtAl2025,
  author  = {Wang, Chen and Ma, Junchi and Wang, Chao and Lin, Defu and Chen, Weilin},
  title   = {Two graphs: Resolving the periodic reversibility of one-dimensional finite cellular automata},
  journal = {Applied Mathematics and Computation},
  volume  = {489},
  pages   = {129151},
  year    = {2025},
  doi     = {10.1016/j.amc.2024.129151},
  eprint  = {2402.05404},
  archivePrefix = {arXiv}
}

@article{NobeYura2004,
  author  = {Nobe, Atsushi and Yura, Fumitaka},
  title   = {On reversibility of cellular automata with periodic boundary conditions},
  journal = {Journal of Physics A: Mathematical and General},
  volume  = {37},
  number  = {22},
  pages   = {5789--5804},
  year    = {2004},
  doi     = {10.1088/0305-4470/37/22/006}
}

@article{Sutner1991,
  author  = {Sutner, Klaus},
  title   = {De Bruijn Graphs and Linear Cellular Automata},
  journal = {Complex Systems},
  volume  = {5},
  number  = {1},
  pages   = {19--30},
  year    = {1991}
}

@misc{OmlandStanica2022,
  author        = {Omland, Tron and St\u{a}nic\u{a}, Pantelimon},
  title         = {Permutation rotation-symmetric S-boxes, liftings and affine equivalence},
  year          = {2022},
  eprint        = {2203.00778},
  archivePrefix = {arXiv},
  primaryClass  = {math.CO},
  note          = {IACR Cryptology ePrint Archive, Report 2022/279}
}
\end{document}